\documentclass[10pt,a4paper]{article}

\usepackage[linesnumbered,ruled,vlined,noresetcount]{algorithm2e}
\DontPrintSemicolon
\SetArgSty{textnormal}
\SetKwProg{Fn}{function}{:}{end}
\SetKwProg{pFor}{for each}{ do in parallel}{end}
\SetKwInput{VarAgent}{Variables in Agent Memory}
\SetKwInput{VarNode}{Variables in $v$'s Whiteboard}
\SetKwInput{Notation}{Notation}
\SetKwInput{Note}{Note}

\usepackage{multirow}

\usepackage{authblk}
\usepackage{geometry}
\usepackage{amsmath,amssymb,amsthm}
\usepackage{graphicx}
\usepackage{hyperref}
\newtheorem{corollary}{Corollary}
\newtheorem{observation}{Observation}
\newtheorem{lemma}{Lemma}
\newtheorem{theorem}{Theorem}

\theoremstyle{definition}

\newcommand{\mcell}[1]{
\begin{tabular}{c}
#1
\end{tabular}
}

\newcommand{\calp}{\mathcal{P}}
\newcommand{\calm}{\mathcal{M}}
\newcommand{\calw}{\mathcal{W}}
\newcommand{\calg}{\mathcal{G}}

\newcommand{\calr}{\mathcal{R}}
\newcommand{\cald}{\mathcal{D}}

\newcommand{\pran}[1]{\calr_{#1}}
\newcommand{\pdet}[1]{\cald_{#1}}

\newcommand{\goforward}{\mathtt{GoForward}}
\newcommand{\circulate}{\mathtt{Circulate}}
\newcommand{\expand}{\mathtt{Expand}}
\newcommand{\initialize}{\mathtt{Initialize}}
\newcommand{\move}{\mathtt{MoveInTree}}

\newcommand{\lastport}{\mathtt{last}}

\newcommand{\clr}{\mathtt{color}}
\newcommand{\parent}{\mathtt{parent}}
\newcommand{\level}{\mathtt{level}}
\newcommand{\phase}{\mathtt{phase}}

\newcommand{\last}{\mathtt{last}}
\newcommand{\stage}{\mathtt{stage}}
\newcommand{\error}{\mathtt{error}}
\newcommand{\port}{\mathtt{port}}
\newcommand{\found}{\mathtt{expanded}}

\newcommand{\firstchild}{\mathtt{firstChild}}
\newcommand{\nextchild}{\mathtt{nextChild}}
\newcommand{\mode}{\mathtt{mode}}

\newcommand{\ind}{\mathtt{index}}

\newcommand{\uts}{\mathtt{UTS}}

\newcommand{\suxs}{\mathtt{SUXS}}

\newcommand{\tr}{\mathit{True}}
\newcommand{\fl}{\mathit{False}}

\newcommand{\self}{\mathit{self}}
\newcommand{\vst}{r}
\newcommand{\pin}{p_\mathrm{in}}
\newcommand{\pout}{p_\mathrm{out}}
\newcommand{\vcur}{v_\mathrm{cur}}
\newcommand{\vtarget}{V_\mathrm{target}}

\newcommand{\nextr}{\mathit{next}_{\calr}}

\newcommand{\nextd}{\mathit{next}_{\cald}}
\newcommand{\colorupdate}{\mathit{colorUpdate}}
\newcommand{\dist}{\mathit{dist}}

\newcommand{\ie}{\textit{i.e., }}

\newcommand{\ex}{\mathbf{E}}

\newcommand{\poly}{\mathit{poly}}
\newcommand{\dmin}{\delta_{\mathrm{min}}}
\newcommand{\dmax}{\delta_{\mathrm{max}}}
\newcommand{\rctime}{O\left( m \cdot  \min\left(D, \frac{n}{c}+1, \frac{D}{c} + \log n \right ) \right)}
\newcommand{\rctimeinline}{O( m \cdot  \min(D, \allowbreak\frac{n}{c}+1, \allowbreak \frac{D}{c} + \log n) )}

\newcommand{\aspace}{\vspace{10pt}}
\newcommand{\varspace}{\hspace{10pt}}

\title{Self-stabilizing Graph Exploration by a Single Agent} 
\date{}

\author[1]{Yuichi Sudo\thanks{Corresponding Author: sudo@hosei.ac.jp}}
\author[2]{Fukuhito Ooshita}
\author[3]{Sayaka Kamei}

\affil[1]{Hosei University, Tokyo, Japan}
\affil[2]{Fukui University of Technology, Fukui, Japan}
\affil[3]{Hiroshima University, Hiroshima, Japan}

\begin{document}

\maketitle

\begin{abstract}
In this paper, we present two self-stabilizing algorithms that enable a single (mobile) agent to explore graphs. Starting from any initial configuration, \ie regardless of the initial states of the agent and all nodes, as well as the initial location of the agent, the algorithms ensure the agent visits all nodes. We evaluate the algorithms based on two metrics: the \emph{cover time}, defined as the number of moves required to visit all nodes, and \emph{memory usage}, defined as the storage needed for maintaining the states of the agent and each node.
The first algorithm is randomized. Given an integer $c = \Omega(n)$, its cover time is optimal, \ie $O(m)$ in expectation, and its memory requirements are $O(\log c)$ bits for the agent and $O(\log (c+\delta_v))$ bits for each node $v$, where $n$ and $m$ are the numbers of nodes and edges, respectively, and $\delta_v$ is the degree of node $v$.
For general $c \ge 2$, its cover time is $\rctime$,
where $D$ is the diameter of a graph. The second algorithm is deterministic. It requires an input integer $k \ge \max(D, \dmax)$, where $\dmax$ is the maximum degree of the graph. The cover time of this algorithm is $O(m + nD)$, and it uses $O(\log k)$ bits of memory for both the agent and each node.
\end{abstract}

\section{Introduction}
\label{sec:introduction}
We focus on the \emph{exploration problem} involving a single mobile entity, referred to as an \emph{agent}, within an undirected, simple, and connected graph $G=(V,E)$. The graph is anonymous, \ie nodes lack unique identifiers. The edges incident to each node $v \in V$ are labeled with locally unique \emph{port numbers} $0, 1, \dots, \delta_v-1$, where $\delta_v$ is the degree of node $v$. The agent, functioning as a finite state machine, migrates from node to node via edges.
Upon migrating from node $u$ to node $v$, the agent learns the degree $\delta_v$ and the \emph{incoming port}, \ie the port number assigned to edge $\{u,v\}$ at node $v$. Moreover, the agent can access and modify $v$'s local memory, known as a \emph{whiteboard}. The agent determines its next destination by computing the port number based on the aforementioned information.
Our objective is to enable the agent to visit every node in the graph in as few steps as possible while minimizing the memory usage of both the agent and the whiteboards. This exploration problem, fundamental in the study of
mobile computing entities, has been extensively studied~\cite{PDD+96,PP99,koucky02,YWI+03,IKY09,MT10,SBN15,GSOM20,GSO+21,IW21}.
Exploration algorithms often serve as foundational tools for solving other fundamental problems, such as rendezvous~\cite{Pelc12,TZ14,PP24}, gathering~\cite{DPP14,DFP+20,SKY+20,SMK+23,BDP23}, dispersion~\cite{AM18,KS25,SSKM20,SSN+24}, and gossiping~\cite{MT10,BDP23}. 

In this paper, we tackle the exploration problem under a more challenging setting: \emph{self-stabilizing exploration}. We do not presuppose any specific initial global state (or \emph{configuration}) of the network. This means that at the start of the exploration, (i) the agent's location within $G$ is arbitrary, (ii) the agent's state is arbitrary, and (iii) the content of each whiteboard is arbitrary. The agent is required to visit all nodes in $G$ from any potentially inconsistent configuration. Generally, an algorithm is considered self-stabilizing \cite{dijkstra82} for problem $P$ if it can solve $P$ starting from any configuration.
Self-stabilizing algorithms are capable of handling any type of transient faults, such as temporary memory corruption, making their design both practically and theoretically significant.
We evaluate algorithms using two metrics: cover time, which is the number of moves required to visit all nodes, and memory usage, which includes the storage needed for the state of the agent and the state of each node.

Generally speaking, several studies tackle a variety of problems involving mobile agents in the self-stabilizing setting \cite{BPT07,MT10,ODM17}. 
In this setting, the number of agents in the graph is fixed. In our case (\ie self-stabilizing exploration by a single agent), the number of agents is always exactly one: we do not consider configurations where no agent exists, or where two or more agents are present. Therefore, this setting may be particularly suitable for applications where physical robots operate in a field represented by an undirected graph, and the robots can leave information in some way at each intersection in the field.


Throughout this paper, 
we denote the number of nodes, the number of edges, and the diameter of a graph 
by $n$, $m$, and $D$, respectively. We denote the degree of a node $v$ by $\delta_v$,
and define $\dmin = \min_{v \in V} \delta_v$ and $\dmax = \max_{v \in V} \delta_v$.

\subsection{Related Work}
If randomization is allowed, we can easily solve self-stabilizing exploration using the well-known strategy called the \emph{simple random walk}.
When the agent visits a node $v \in V$, it simply chooses a node as the next destination uniformly at random
among $N(v)$,
where $N(v)$ is the set of all neighbors of $v$ in $G$.
In other words, it moves to any node $u \in N(v)$ with probability $P_{v,u} = 1/\delta_v$.
It is well known that the agent running this simple algorithm visits all nodes in $G$
within $O(\min(mn,mD\log n))$ steps in expectation
where $n=|V|$, $m=|E|$, and $D$ is the diameter of $G$.
 (See \cite{AKJ+79,Mat88}.) 
Since the agent is oblivious (\ie the agent does not bring any information at a migration between two nodes) and does not use whiteboards, the simple random walk is obviously a self-stabilizing exploration algorithm. 

Ikeda, Kubo, and Yamashita \cite{IKY09} improved the cover time
(\ie the number of steps to visit all nodes)
of the simple random walk
by setting the transition probability as
$P'_{v,u} = \delta_u^{-1/2}/\sum_{w \in N(v)}\delta_w^{-1/2}$
for any $u \in N(v)$.
They proved that the cover time of this \emph{biased random walk} is $O(n^2 \log n)$ steps in expectation.
However, we cannot use this result directly in our setting
because the agent must know the degrees of all neighbors of the current node to compute the next destination. 
We can implement this random walk, for example, as follows:
every time the agent visits node $v$,
it first obtains $(\delta_u)_{u \in N(v)}$
by visiting all $v$'s neighbors in $2\delta_v$ steps,
and then decides the next destination 
according to probability $(P'_{v,u})_{u \in N(v)}$,
which is now computable with $(\delta_u)_{u \in N(v)}$.
However, this implementation increases the cover time
by a factor of at least $\dmin$ and at most $\dmax$.
Whereas $n^2 \dmax \log n > mn$ always holds,
$n^2 \dmin \log n < \min(mn,mD\log n)$ may also hold.
Thus, we cannot determine which random walk has smaller cover time
without detailed analysis.
To bound the space complexity, we must know an upper bound
$\Delta$ on $\dmax$ to implement this random walk. 
If the agent stores $(\delta_u)_{u \in N(v)}$ on $v$'s whiteboard,
it uses $O(\log \Delta)$ bits in the agent-memory
and $O(\delta_v \log \Delta)$ bits in the whiteboard of each node $v$.
If the agent stores $(\delta_u)_{u \in N(v)}$
only on the agent-memory,
it uses $O(\Delta \log \Delta)$ bits in the agent-memory.

The algorithm given by Priezzhev, Dhar, Dhar, and Krishnamurthy \cite{PDD+96}, which is nowadays
well known as the \emph{rotor-router},
solves the self-stabilizing exploration
deterministically.
The agent is oblivious, but it uses only $O(\log \delta_v)$
bits in the whiteboard of each node $v \in V$.
The edges $(\{v,u\})_{u \in N(v)}$ are assumed
to be locally labeled by $0,1,\dots,\delta_v-1$ in a node $v$.
The whiteboard of each node $v$ has one variable
$v.\last \in \{0,1,\dots,\delta_v-1\}$.
Every time the agent visits a node $v$,
it increases $v.\lastport$ by one modulo $\delta_v$
and moves to the next node via the edge
labeled by the updated value of $v.\lastport$.
This simple algorithm guarantees that
starting from any configuration, 
the agent visits all nodes within $O(mD)$ steps \cite{YWI+03}.
Masuzawa and Tixeuil \cite{MT10} also gave a deterministic
self-stabilizing exploration algorithm.
This algorithm itself is designed to solve the gossiping problem
where two or more agents have to share their given information with each other.
However, this algorithm has a mechanism to visit all the nodes
starting from any configuration, which can be seen
as a self-stabilizing exploration algorithm.
The cover time and the space complexity for the whiteboards
of this algorithm are asymptotically the same
as those of the rotor-router,
while it uses a constant space of the agent-memory,
unlike oblivious algorithms such as the rotor-router.
Shimoyama, Sudo, Kakugawa, and Masuzawa~\cite{SSKM22} presented a deterministic self-stabilizing exploration algorithm, restricted to cactus graphs.

A \emph{Universal Traversal Sequence} (UTS), introduced by Aleliunas, Karp, Lipton, Lovász, and Rackoff~\cite{AKJ+79}, is also closely related to deterministic self-stabilizing exploration. Let $\calg_{N,d}$ denote the set of (not necessarily simple) connected $d$-regular graphs with at most $N$ nodes. A sequence $p_0, p_1, \dots, p_{\ell-1}$ of integers from $\{0,1,\dots,d-1\}$ is called a UTS for $\calg_{N,d}$ if, for every graph $G \in \calg_{N,d}$, the agent visits all nodes in $G$ by moving through port $p_i$ at each time step $i=0,1,\dots,\ell-1$. Aleliunas et al.~\cite{AKJ+79} proved the existence of a UTS of length $O(N^3 d^2 \log N)$ for $\calg_{N,d}$ using the probabilistic method.
From this result, given a positive integer $N$, we immediately derive a self-stabilizing exploration algorithm for any simple and connected graph of size at most $N$ with a cover time of $O(N^5 \log N)$, as follows: Let $P = p_0, p_1, \dots, p_{\ell-1}$ be a UTS of length $O(N^5 \log N)$ for $\calg_{N,N-1}$. The agent maintains a single variable $\ind \in \{0,1,\dots,\ell-1\}$, and at each time step, moves through port $p_{\ind}$ and increments $\ind$ by one modulo $\ell$. The agent can compute the UTS locally at each time step, thus requiring only $O(\log N)$ bits of agent memory to maintain $\ind$.
One might be concerned that the above UTS is designed specifically for regular graphs, thus questioning its applicability to arbitrary graphs. However, this restriction is not significant, as any arbitrary graph can be virtually transformed into a regular graph by adding self-loops (see \cite{TZ14} for details).

Kouck{\`y}~\cite{koucky02} introduced a similar concept called a \emph{Universal Exploration Sequence} (UXS). Although we omit the definition of a UXS here, given a UXS, we can immediately derive a self-stabilizing exploration algorithm in a similar way by using the incoming port information. Reingold~\cite{Reingold08} proved that, given positive integers $N$ and $d$, a UXS of length $\mathit{poly}(N)$ for $\calg_{N,d}$ can be explicitly constructed in log-space and, hence, in polynomial time. (Note that the proof by Aleliunas et al.~is non-constructive, and thus local computation of a UTS may require super-exponential time.) Later, Ta-Shma and Zwick~\cite{TZ14} introduced the concept of a \emph{Strongly Universal Exploration Sequence} (SUXS) and obtain similar results, leading to a cover time of $O(n^{10} \log^2 n)$. Compared to a UTS, the exponent is doubled, but this cover time no longer depends on a given upper bound $N$; instead, it depends only on the actual size $n$. Xin~\cite{Xin07} improved this cover time to $O(n^{5+\varepsilon})$, where $\varepsilon$ is an arbitrarily small positive constant.


A few self-stabilizing algorithms were given for mobile agents
to solve problems other than exploration.
Blin, Potop-Butucaru, and Tixeuil \cite{BPT07}
studied the self-stabilizing naming and leader election problem.
Masuzawa and Tixeuil \cite{MT10} gave a self-stabilizing
gossiping algorithm. Ooshita, Datta, and Masuzawa \cite{ODM17}
gave self-stabilizing rendezvous algorithms.

If we assume a specific initial configuration,
that is, if we do not require a self-stabilizing solution,
the agent can easily visit all nodes within $2m$ steps
with a simple depth first traversal (DFT).
Panaite and Pelc \cite{PP99} gave a faster algorithm,
whose cover time is $m + 3n$ steps. 
They assume that the nodes are labeled by the unique identifiers.
Their algorithm uses $O(m \log n)$ bits in the agent-memory,
while it does not use whiteboards.
Sudo, Baba, Nakamura, Ooshita, Kakugawa, and Masuzawa \cite{SBN15}
gave another implementation of this algorithm:
they removed the assumption of the unique identifiers
and reduced the space complexity on the agent-memory
from $O(m\log n)$ bits to $O(n)$ bits by using $O(n)$ bits in
each whiteboard.
It is worthwhile to mention that
these algorithms \cite{PP99,SBN15}
guarantee the termination of the agent
after exploration is completed,
whereas designing a self-stabilizing exploration algorithm
with termination is impossible.
Self-stabilization and termination
contradict each other by definition:
if an agent-state that yields termination exists,
the agent never completes exploration
when starting exploration with this state. 
If such state does not exist,
the agent never terminates the exploration.

In the classical or standard distributed computing model (excluding mobile agents), the self-stabilizing token circulation problem, particularly self-stabilizing depth-first token circulation (DFTC), has been extensively studied \cite{HC93,DJPV00,Petit01}. Introduced by Huang and Chen \cite{HC93}, this problem was addressed with a self-stabilizing DFTC algorithm using $O(\log n)$ bits per process, which was later reduced to $O(\log \dmax)$ bits by Datta, Johnen, Petit, and Villain \cite{DJPV00}. Petit \cite{Petit01} developed a time-optimal (\ie $O(n)$-time) self-stabilizing DFTC algorithm that also requires $O(\log n)$ bits per process.
However, these algorithms are not directly applicable to self-stabilizing exploration by a single agent because the network models are fundamentally different. In the standard model, $n$ computational processes can communicate with each other via communication links in parallel, whereas in our model, only a single agent computes and updates the states of nodes in the network, potentially requiring more time to solve problems.
On the other hand, one of the main challenges for self-stabilizing token circulation in the standard model is maintaining exactly one token starting from any configuration where there may be no tokens or where two or more tokens may exist. As mentioned above, this challenge does not apply to our model, where there is always a single agent in any configuration.
Yet, many techniques from standard distributed computing might be adaptable for mobile agent algorithms. For example, our self-stabilizing exploration algorithms employ the technique of repeatedly recoloring graph nodes to resolve variable inconsistencies, a common approach in the design of self-stabilizing algorithms (see Dolev, Israeli, and Moran \cite{DIM97}).

\begin{table*}[t]
\caption{Randomized self-stabilizing graph exploration algorithms for a single agent.
}
\label{tbl:rand}
\center
{
\begin{tabular}{c c c c}
\hline
 & Expected Cover Time & \mcell{Agent \\Memory} & \mcell{Memory on \\ node $v$}\\
\hline
Simple Random Walk & $O(\min(mn,mD\log n))$ & 0 & 0\\
\mcell{Biased Random Walk \cite{IKY09}\\ (require $\Delta \ge \dmax$ )}&
$O(n^2 \dmax \log n)$ & \mcell{$O(\log \Delta)$\\$O(\Delta \log \Delta)$} & \mcell{$O(\delta_v \log \Delta)$\\$0$}\\
$\pran{c}$ (require $c\ge 2$) &
$\rctime{}$
& $O(\log c)$ &
$O(\log (\delta_v+c))$
\\
\hline
\end{tabular}
}
\vspace{10pt}
\caption{Deterministic self-stabilizing graph exploration algorithms for a single agent,
where $\varepsilon > 0$ (fourth row) is an arbitrary small constant.
}
\label{tbl:det}
\center
{
\begin{tabular}{c c c c}
\hline
 & Cover Time & \mcell{Agent \\ Memory} &\mcell{Memory on \\ node $v$}\\
\hline 
Rotor-router \cite{PDD+96} & $O(mD)$ & 0 & $O(\log \delta_v)$ \\
 $\uts_{N}$~\cite{AKJ+79} (require $N \ge n$)
 & $O(N^5 \log N)$& $O(\log N)$ & 0 \\
$\suxs_N$~\cite{TZ14} (require $N \ge n$) & $O(n^{10}\log^2 n)$ & $O(\log N)$ & 0 \\
$\suxs_N$~\cite{Xin07} (require $N \ge n$) & $O(n^{5+\varepsilon})$ & $O(\log N)$ & 0 \\
2-color DFT \cite{MT10}
& $O(mD)$ & $O(1)$ & $O(\log \delta_v)$ \\
$\pdet{k}$ (require $k \ge \max(D,\dmax)$)  & $O(m+nD)$ & $O(\log k)$ & $O(\log k)$ \\
\hline
\end{tabular} 
}
\end{table*}

\subsection{Our Contribution}

In this paper, we investigate the minimum achievable cover time in a self-stabilizing setting. It is easy to see that the cover time is lower bounded by $\Omega(m)$: any (even randomized) algorithm requires at least $\Omega(m)$ steps in expectation before the agent visits all nodes. For completeness, we prove this lower bound in the appendix (Section \ref{sec:lower}). Our goal is to design an algorithm whose cover time is close to this lower bound while minimizing the required agent memory and whiteboard space.

We give two self-stabilizing exploration algorithms
$\pran{c}$ and $\pdet{k}$,
where $c$ and $k$ are the design parameters.
The cover times and the space complexities
of the proposed algorithms and the existing algorithms
are summarized in Tables \ref{tbl:rand} and \ref{tbl:det}.

Algorithm $\pran{c}$ is a randomized algorithm that achieves a cover time of $\rctimeinline$ steps in expectation, utilizes $O(\log c)$ bits of agent memory, and requires $O(\log (\delta_v + c))$ bits for the whiteboard of each node $v$.
Thus, we have trade-off between the cover time and
the space complexity.
The larger $c$ we use, the smaller cover time we obtain.
In particular, 
the expected cover time
is $O(m\log n)$ steps if we set $c=\Omega(D/\log n)$,
and it becomes optimal (\ie $O(m)$ steps) if we set
$c = \Omega(n)$.
This means that we require the knowledge of
$\Omega(n)$ value to make $\pran{c}$ time-optimal.
Fortunately, this assumption can be ignored
from a practical point of view:
even if $c$ is extremely larger than $n$, 
the overhead will be just an additive factor of $\log c$
in the space complexity.
Thus, any large $c \in \poly(n) \cap \Omega(n)$ is enough to obtain
the optimal cover time and the space complexity of
$O(\log n)$ bits both in the agent memory and whiteboards.
Moreover, irrespective of parameter $c \ge 2$,
the cover time is $O(mD)$ steps with probability $1$.

Algorithm $\pdet{k}$ is deterministic, and its cover time is $O(m+nD)$ steps, independent of the parameter $k$. The agent uses $O(\log k)$ bits both for agent memory and the whiteboard at each node. Thus, we do not have a trade-off between cover time and space complexity. Unlike $\pran{c}$, however, algorithm $\pdet{k}$ requires prior knowledge of an upper bound on the diameter and maximum degree, specifically, it assumes $k \ge \max(D,\dmax)$. If this assumption is not met, the correctness of the algorithm is no longer guaranteed.
Nevertheless, this assumption can still be considered weak: since the space complexity grows only logarithmically in $k$, it is always feasible to choose a sufficiently large polynomial upper bound on $n$ for $k$ without significantly affecting memory usage. Thus, even though the algorithm formally requires knowledge of an upper bound, this requirement does not substantially limit the applicability of this algorithm.

It remains open if there is a deterministic self-stabilizing exploration algorithm with optimal cover time, \ie $O(m)$ steps.

\section{Preliminaries}
\label{sec:model}
Let $G=(V,E,p)$ be a simple, undirected, and  connected graph
where $V$ is the set of nodes and
$E$ is the set of edges.
The edges are locally labeled at each node:
we have a family of functions
$p=(p_v)_{v\in V}$
such that
each $p_v: \{\{v,u\} \mid u \in N(v)\} \to \{0,1,\dots,\delta_v-1\}$ uniquely assigns
a \emph{port number} to every edge incident to node $v$.
Two port numbers $p_u(e)$ and $p_v(e)$ are
independent of each other for edge $e =\{u,v\}\in E$.
The node $u$ neighboring to node $v$ such that
$p_v(\{v,u\}) = q$ is called the \emph{$q$-th neighbor} of $v$
and is denoted by $N(v,q)$. 
For any two nodes $u,v$, we define the distance between $u$ and $v$ as the length of a shortest path between them and denote it by $d(u,v)$. We also define the set of the \emph{$i$-hop neighbors} of $v$ as $N_i(v) = \{u \in V \mid d(u,v) \le i\}$.
Note that $N_0(v)=\{v\}$ and $N_1(v)=N(v) \cup \{v\}$.


An \emph{algorithm} 
is defined as a 3-tuple $\calp=(\phi,\calm,\calw)$,
where $\calm$ is the set of states for the agent,
$\calw=(\calw_k)_{k \in \mathbb{N}}$
is the family such that $\calw_k$ is
the set of states for each node with degree $k$,
and $\phi$ is a function that determines
how the agent updates its state (\ie agent memory) and
the state of the current node (\ie \emph{whiteboard}). 
At each time step,
the agent is located at exactly one node
$v \in V$,
and moves through an edge incident to $v$.
Every node $v \in V$ has a whiteboard $w(v) \in \calw_{\delta_v}$,
which the agent can access freely when it visits $v$.
The function $\phi$ is invoked
every time the agent visits a node or when the exploration begins.
Suppose that the agent with state $s \in \calm$
has moved to node $v$ with state $w(v) = x \in \calw_{\delta_v}$
from $u \in N(v)$. Let $\pin=p_v(\{u,v\})$.
The function $\phi$ takes 4-tuple $(\delta_v,\pin,s,x)$ as the input
and returns 3-tuple $(\pout,s',x')$ as the output.
Then, the agent updates its state to $s'$
and $w(v)$ to $x'$, after which it moves via port $\pout$,
that is, it moves to $v'$ such that $\pout = p_v(\{v,v'\})$.
At the beginning of an execution, we let $\pin$ be
an arbitrary integer in $\{0,1,\dots,\delta_v-1\}$ where
$v$ is the node at which the agent is located. 
Note that if algorithm $\calp$ is randomized one,
function $\phi$ returns the probabilistic distributions
for $(\pout,s',x')$.



Given a graph $G=(V,E)$, a \emph{configuration} (or a global state)
consists of the location of the agent, the state of the agent
(including $\pin$), and the states of all the nodes in $V$.
Algorithm $\calp$ is a self-stabilizing exploration algorithm
for a class $\calg$ of graphs
if for any graph $G=(V,E,p) \in \calg$, the agent running $\calp$ on $G$ eventually visits all the nodes in $V$ at least once
starting from any configuration.
Note that, by the above definition, any self-stabilizing exploration algorithm ensures that the single agent visits every node infinitely often.

We measure the efficiency of algorithm $\calp$
by three metrics:
\emph{the cover time}, \emph{the agent memory space},
and \emph{the whiteboard memory space}.
All the above metrics are evaluated in the worst-case manner
with respect to graph $G$ and an initial configuration.
The cover time is defined as the number of moves
that the agent makes before it visits all nodes.
If algorithm $\calp$ is a randomized one,
the cover time is evaluated in expectation. 
The memory spaces of the agent and the whiteboard on node $v$
are just defined
as $\log_2 |\calm|$ and $\log_2|\calw_{\delta_v}|$,
respectively.

Throughout this paper, we refer to the node where the agent is currently located as \emph{the current node}, denoted by $\vcur$. We call the port number of the current node via which
the agent migrates to $\vcur$ \emph{in-port}
and denote it by $\pin$.
As mentioned above, the function $\phi$ can use $\pin$
to compute 3-tuple $(\pout, s', x')$.
In other words, the agent can always access $\pin$
to update its variables and
the whiteboard variables of the current node
and compute the destination of the next migration.

\paragraph*{\textbf{Algorithm Description}}

This paper presents two algorithms.
For simplicity,
instead of giving a formal 3-tuple $(\phi,\calm,\calw)$,
we specify the set of agent variables,
the set of whiteboard variables,
and the pseudocode of instructions that the agent performs.
In addition to the agent variables, 
the agent must convey the program counter of the pseudocode
and the call stack (or the function stack) when its migrates
between nodes to execute an algorithm consistently.
Thus, all the agent variables, the program counter,
and the call stack
constitute the state of the agent.
Conforming to the above formal model,
the state of the agent changes only when it migrates
between nodes. 
Therefore, the program counter of each state
is restricted to one of the line numbers
corresponding to the instructions of migration.
For example, in Algorithm \ref{al:pran} explained
in Section \ref{sec:random}, the domain of the program counter
in the agent states are $\{13,17,19\}$.
We can ignore the space for the program counter 
and the call stack to evaluate the space complexity
because they require only $O(1)$ bits.
(Our pseudocodes do not use a recursive function.)
Whiteboard variables are stored in the whiteboard $w(v)$
of each node $v$.
All the whiteboard variables in $w(v)$
constitute the state of $w(v)$.
 For clarification,
 we denote an agent variable $x$ by $\self.x$
 and an whiteboard variable $y$ in $w(v)$ by $v.y$.

\section{Randomized Exploration}
\label{sec:random}
In this section, we give a randomized self-stabilizing exploration algorithm $\pran{c}$. The cover time and space complexity of $\pran{c}$ depend on the design parameter $c \geq 2$. Specifically, the cover time is $O(m)$ if $c = \Omega(n)$, while it is $O(m \log n)$ if $c = \Omega(D/\log n)$. The cover time is $O(mD)$ with probability 1, regardless of how small $c$ may be. Algorithm $\pran{c}$ requires $O(\log c)$ bits of agent memory and $O(\log (\delta_v+c))$ bits for each node $v \in V$. If $c = 2$, the algorithm becomes deterministic and is nearly identical to the exploration algorithm given by \cite{MT10}.

The idea of algorithm $\pran{c}$ is simple.
The agent tries to traverse all the edges in every $2m$ moves according to the Depth-First-Traversal (DFT).
However, we have an issue for executing DFT.
When the agent visits a node for the first time,
it has to mark the node as an \emph{already-visited} node.
However, 
we must deal with an arbitrary initial configuration
in a self-stabilizing setting,
thus some of the nodes may be marked even at the beginning of an execution. 
To circumvent this issue, we use \emph{colors} and let the agent make DFTs repeatedly.
The agent and all nodes maintain
their color in a variable $\clr \in \{1,2,\dots,c\}$.
Every time the agent begins a new DFT, 
it chooses a new color $\rho$ uniformly at random from all the colors except for the current color of the agent.
In the new DFT,
the agent changes the colors of all the visited nodes to $\rho$.
Thus, the agent can distinguish the visited nodes and 
the unvisited nodes with their colors:
it interprets that the nodes colored $\rho$ are already visited and
the nodes having other colors have not been visited before
in the current DFT.
Of course, the agent may make incorrect decision
because one or more nodes may be colored $\rho$
when the agent chooses $\rho$ for its color at the beginning of a new DFT.
Thus, the agent may not perform a complete DFT. 
However, the agent can still make a progress for exploration
to a certain extent in this DFT:
Let $\vst$ be the node at which the agent is located 
when it begins the DFT,
let $S$ be the set of the nodes colored $\rho$ at that time,
and let $R$ be the set of all nodes in the connected component including $\vst$ of the induced subgraph $G[V \setminus S]$;
then it visits all nodes in $R$
and their neighbors.
Hence, as we will see later, the agent visits
all nodes in $V$ 
by repeating DFTs a small number of times,
which depends on how large $c$ is.

\begin{figure}[t]
  \centering
  \includegraphics[width=0.55\linewidth]{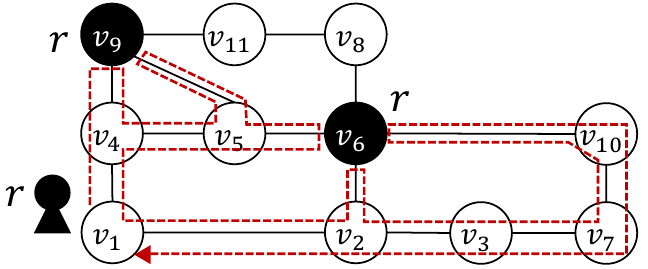}
  \caption{An example of the trajectory of one DFT}
 \label{fig:pran}
\end{figure}

Figure \ref{fig:pran} shows an example of one DFT where
$\vst = v_1$ and $S=\{v_6,v_9\}$. The agent tries to make a DFT
on the graph starting from node $v_1$.
Whenever the agent visits node $v_6$ or $v_9$, which are colored $\rho$,
it immediately \emph{backtracks} to the previous node.
This is because it mistakenly \emph{thinks} that it has already visited these nodes even when it visits them for the first time
in the DFT.
Thus, it never visits $v_8$ or $v_{11}$ in this DFT, which is unreachable
from node $v_1$ in the induced subgraph $G[V \setminus S]$.
However, it visits all the other nodes.
Note that the index of the nodes
are depicted in the figure only for simplicity of explanation,
and we do not assume the existence of any unique identifiers of the nodes.

\begin{algorithm}
\caption{$\pran{c}$}
\label{al:pran}
\Notation{$\nextr(v) = (\pin + 1) \bmod \delta_{v}$}
\aspace
\While{$\tr$}{
Choose $\self.\clr$ uniformly at random from $\{1,2,\dots,c\} \setminus \{\self.\clr\}$\; 
$\vcur.\clr \gets \self.\clr$\;
$\goforward(0)$\;
\While{$\lnot (\vcur.\parent = \bot \land \pin = \delta_{\vcur}-1)$}{
\uIf{$\vcur.\clr \neq \self.\clr$}{
$\vcur.\clr \gets \self.\clr$\; $\vcur.\parent \gets \pin$\;
$\goforward(\nextr(\vcur))$\;
}
\Else{
\uIf{$\vcur.\parent = \nextr(\vcur)$}{
$\vcur.\parent \gets \bot$\;
Migrate to $N(\vcur,\nextr(\vcur))$\;
\tcp*{type-II backtracking}
}
\Else{
$\goforward(\nextr(\vcur))$\;
}
}
}
}
%
\Fn{$\goforward(q)$}{
Migrate to node $N(\vcur,q)$ \tcp*{forward move}
\If{$\vcur.\clr = \self.\clr$}{
Migrate to node $N(\vcur,\pin)$
\tcp*{type-I backtracking}
}
}
\end{algorithm}

\subsection{Detailed Algorithm Description}
The pseudocode of $\pran{c}$ is shown in Algorithm $\ref{al:pran}$. The agent maintains one variable $\self.\clr \in \{1,2,\dots,c\}$, while each node $v$ maintains two variables
$v.\parent \in \{\bot,0,1,\dots,\allowbreak \delta_v-1\}$
and $v.\clr \in \{1,2,\dots,c\}$.
As mentioned before,
the agent uses variables $\self.\clr$ and $v.\clr$
to distinguish already-visited nodes and unvisited nodes
in the current DFT. It also uses whiteboard variable
$v.\parent$ to go back to the node at which it began
the current DFT.
Specifically, whenever the agent finds a node having different color from $\self.\clr$, it changes $\vcur.\clr$ to $\self.\clr$ while it simultaneously substitutes $\pin$ for $\vcur.\parent$ (Lines 7,8).
Thus, the agent performs a DFT with making a spanning tree on graph $G$. We say that node $u$ is the parent of node $v$ when $u = N(v,v.\parent)$.

In an execution of this algorithm, there are three kinds of agent-migration between nodes:
\begin{description}
 \item \textbf{Forward Moves}: The agent executes this kind of migration when it tries to find an unvisited node in the current DFT,
that is, a node not colored $\self.\clr$.
The agent makes a forward move only in Line 17 in function
$\goforward()$, which is invoked in Lines 4, 9, and 15.
\item \textbf{Backtracking (Type I)}: 
The agent executes this kind of migration when the agent makes a forward move, but the destination node already has color $\self.\clr$. By this migration, the agent backtracks to the previous node at which it made a forward move (Line 19). 
\item \textbf{Backtracking (Type II)}:
The agent executes this kind of migration
when it thinks that it has already visited all nodes in $N(\vcur)$. 
Specifically, it backtracks to the parent of $\vcur$
when $\vcur.\parent = (\pin + 1) \bmod \delta_{v}$  (Line 13).
\end{description}

Every time the agent begins a new DFT,
it chooses a new color uniformly at random from $\{1,2,\dots,c\}\setminus \{\self.\clr\}$ (Line 2).
At this time, the starting node of this DFT, say $\vst$, satisfies $\vst.\parent = \bot$ because the while-loop at Lines 5-15 ends if and only if $\vcur.\parent = \bot$ holds. 
Thereafter, it performs a DFT with the new color $\self.\clr$
(Lines 3-15).
The agent first changes $\vst.\clr$ to $\self.\clr$
and makes a forward move to node $N(\vst,0)$ (Lines 3-4).
Thereafter, whenever it finds an unvisited node,
it marks the node with color $\self.\clr$ and 
keeps on making a forward move (Lines 6-9).
If the destination node of a forward move already has color $\self.\clr$, it makes type-I backtracking. 
After a (type-I or type-II) backtracking, 
the agent decides whether there are still unvisited neighbors in $N(\vcur)$. The agent can make this decision according to the predicate $\vcur.\parent = \nextr(\vcur)$ (Line 11) because
it must have made forward moves to
$N(v,(\vcur.\parent+1) \bmod \delta_v), N(v,(\vcur.\parent+2) \bmod \delta_v),\allowbreak \dots, N(v,(\vcur.\parent + \delta_v -1) \bmod \delta_v)$ 
in this order if this predicate holds.

If this predicate does not hold, the agent makes a forward move to node $N(\vcur,\nextr(\vcur))$ (Line 15). Otherwise, it performs type-II backtracking (Line 13). Notably, during this process, the agent simultaneously updates $\vcur.\parent$ to $\bot$ (Line 12). This update is necessary to prevent endless type-II backtracking that could occur without this update when starting from a configuration where $(\{v,v.\parent\})_{v \in V}$ has a cycle.
The agent eventually returns to $\vst$ from node $N(\vst,\delta_{\vst}-1)$. Then, it terminates the current DFT (\ie breaks the while-loop at Lines 5-15) and begins a new DFT.

\subsection{Correctness and Cover Time}

\begin{lemma}
 \label{lemma:back}
 If the agent running $\pran{c}$ makes a forward move from node $v \in V$
 to node $u \in N(v)$, 
 thereafter no type-II backtracking to $v$ occurs before
 the agent makes type-II backtracking from $u$ to $v$.
\end{lemma}
\begin{proof}
By the definition of $\pran{c}$, the set 
$(\{v,v.\parent\})_{v\in V}$ always contains
a path from $\vcur$ to $v$ during the period from the forward move to the type-II backtracking from $u$ to $v$.
Therefore, the agent never makes type-II backtracking from any $w \in N(v) \setminus \{u\}$ to $v$ during the period.
\end{proof}

\begin{lemma}
\label{lemma:DFTtime}
Starting from any configuration,
the agent running $\pran{c}$ changes its color (\ie begins a new DFT) within $8m+n=O(m)$ rounds
with probability 1.
\end{lemma}

\begin{proof}
Let $C$ be an any configuration and let $\rho$ be the color of the agent (\ie $\rho=\self.\clr$) in $C$.
Let $M_F(v)$, $M_{B1}(v)$, and $M_{B2}(v)$ be the numbers of forward moves,
type-I backtracking, and type-II backtracking
that the agent makes at $v \in V$, respectively,
until it changes $\self.\clr$ from $\rho$
in the execution starting from $C$.
Whenever the agent makes type-II backtracking at $v$,
it changes $v.\parent$ to $\bot$.
Thus, it never makes type-II backtracking twice at $v$,
that is, $M_{B2}(v) \le 1$.
Next, we bound $M_{F}(v)$. 
By Lemma \ref{lemma:back},
once the agent makes a forward move from $v$ to $N(v,i)$, 
the following migrations involving $v$ must be
type-II backtracking from $N(v,i)$ to $v$, 
the forward move from $v$ to $N(v,i+1)$,
type-II backtracking from $N(v,i+1)$ to $v$, 
the forward move from $v$ to $N(v,i+2)$, and so on.
Therefore, the agent makes type-II backtracking at $v$ 
or observes $v.\parent = \bot \land \pin = \delta_v$
(and changes its color) before it makes forward moves
$\delta_v$ times. Since the agent makes 
type-II backtracking at $v$ at most once, 
we also have $M_{F}(v) \le 2\delta_v$.
Type-I backtracking occurs only after the agent makes a forward move. Therefore, we have $\sum_{v \in V}M_{B1}(v) \le \sum_{v \in V}M_{F}(v)$.
To conclude, we have $\sum_{v \in V}(M_{F}(v)+M_{B1}(v)+M_{B2}(v)) \le \sum_{v\in V} 2M_{F}(v) + \sum_{v \in V} M_{B2}(v) = 8m + n$.
\end{proof}

\begin{lemma}
\label{lemma:expand}
Suppose now that the agent running $\pran{c}$
changes its color to $\rho$ and begins a new DFT at node $\vst$.
Let $S \subseteq V$ be the set of the nodes colored $\rho$ at that time
and let $R$ be the set of the nodes
in the connected component including $\vst$ in the induced subgraph $G[V \setminus S]$.
Then, the agent visits all nodes in $R$
and their neighbors in this DFT with probability 1.
Moreover, this DFT finishes (\ie the agent changes its color)
at node $\vst$.
\end{lemma}

\begin{proof}   
By the definition of $\pran{c}$, the set 
$(\{v,v.\parent\})_{v\in V}$ always contains
a path from $\vcur$ to $\vst$.
Therefore, this DFT terminates only at $\vst$.
Thus it suffices to show that the agent
visits all nodes in 
$R \cup \bigcup_{v \in R} N(v)$ in this DFT.
Let $u$ be any node in $R \cup \bigcup_{v \in R} N(v)$.
By definition of $R$, there exists a path
$v_0,v_1,v_2,\dots,v_l$ in $G$ where $\vst=v_0$, $u=v_l$ and $v_i \notin S$ for all $i = 1,2,\dots,l-1$.
As mentioned in the proof of Lemma \ref{lemma:DFTtime},
once the agent makes a forward move from $v$ to $N(v,i)$, 
the following migrations involving $v$ must be
type-II backtracking from $N(v,i)$ to $v$, 
the forward move from $v$ to $N(v,i+1)$,
type-II backtracking from $N(v,i+1)$ to $v$, 
the forward move from $v$ to $N(v,i+2)$, and so on.
Moreover, the agent makes the first forward move at $\vst$ to $N(\vst,0)$ at Line 4 and
makes the first forward move at each $v_i$ for $i=1,\dots,l-1$ to $N(v_i,(v_i.parent + 1) \bmod \delta_{v_i})$ at Line 9
because we have $v_i \notin R$ for $i=0,1,\dots,l-1$.
This means that for each $i=0,1,\dots,l-1$,
the agent makes a forward move exactly once
from $v_i$ to each $u \in N(v_i)$ except for $v_i$'s parent
in the current DFT. 
In particular, the agent makes a forward move from $v_{l-1}$
to $v_{l}=u$, thus the agent visits $u$ in this DFT.
\end{proof}



In the DFT mentioned in Lemma \ref{lemma:expand},
the agent changes the colors of all nodes in $R'=R \cup \bigcup_{v \in R}N(v)$ to $\rho$.
In the next DFT, the agent must choose a color different from $\rho$,
ensuring that the agent will visit all nodes in $R'$ and their neighbors (plus possibly many more nodes).
Consequently, the number of the nodes visited in the $i$-th DFT monotonically increases with respect to $i$.
Therefore, we obtain the following corollary.

\begin{corollary}
\label{cor:visitall}
The agent running $\pran{c}$ visits all nodes
before it changes its color $D+1$ times.
\end{corollary}

\begin{theorem}
\label{theorem:pran}
The following propositions hold:
(i) Algorithm $\pran{c}$ is a randomized
self-stabilizing exploration algorithm
for all simple, undirected, and connected graphs,
(ii) irrespective of $c$, the cover time is $O(mD)$ steps
with probability $1$,
(iii) the expected cover time is 
$\rctimeinline$ steps, and
(iv) the agent memory space is $O(\log c)$
and the memory space of each node $v$ is $O(\log (\delta_v+c))$.
\end{theorem}

\begin{proof}
The first and the second claims of this theorem
immediately follow from Lemma \ref{lemma:DFTtime}
and Corollary \ref{cor:visitall}.
The fourth claim holds because the agent maintains exactly one color (from the set $\{1,2,\dots,c\}$) in its memory and on the whiteboard of each node, and stores the parent pointer of each node $v$ (from the set $\{0,1,\dots,\delta_v\}$) on the whiteboard of each node $v$,
and do not use any other variable with a super-constant size.
In the rest of this proof,
we prove the third claim of this lemma:
the agent visits all nodes within
$\rctimeinline$ steps in expectation.
Let $X$ denote how many times the agent changes its color
before it visits all nodes in $V$. 
By Lemma \ref{lemma:DFTtime} and $\Pr[X \le D+1]=1$,
it suffices to show
$\ex[X] = O(n/c +1)$
and $\ex[X] = O(D /c + \log n)$.
We exclude the case $c = 2$ without loss of generality because
these equalities immediately follow from $X \le D+1$ when $c=2$.

First, we prove that $\ex[X] = O(n/c + 1)$.
For each integer $i \ge 1$, we define the $i$-th DFT as the DFT executed by the agent immediately after it has changed its color exactly $i$ times. Let $V_i$ be the set of nodes visited by the agent during the $i$-th DFT. Fix some $i \ge 1$ and let $V' \subseteq V$ be an arbitrary superset of $V_i$ satisfying $|V' \setminus V_i| = \lfloor c/2 \rfloor$, with the induced subgraph of $V'$ being connected. If $|V_i| + \lfloor c/2 \rfloor > n$, we define $V' = V$.
By Lemma \ref{lemma:expand}, if the color chosen by the agent for the $(i+1)$-th DFT differs from the colors of all nodes in $V' \setminus V_i$, then the agent visits all nodes of $V'$ during the $(i+1)$-th DFT. In that case, we have $|V_{i+1}| \ge \min(|V_i| + \lfloor c/2 \rfloor, n)$. Such an event occurs with probability
$$
p \ge 1 - \frac{|V' \setminus V_i|}{c - 1} \ge 1 - \frac{\lfloor c/2 \rfloor}{c - 1} = \Theta(1),
$$
assuming $c \ge 3$.
Thus, in expectation, $|V_i|$ increases by at least $\lfloor c/2 \rfloor$ or reaches $n$ within $1/p = O(1)$ DFTs, until $V_i = V$ holds. This yields
$$
\ex[X] \le \left\lceil \frac{n}{\lfloor c/2 \rfloor} \right\rceil \cdot \frac{1}{p} = O\left(\frac{n}{c} + 1\right).
$$

Next, we prove $\ex[X] = O(D /c  + \log n)$.
By Lemma \ref{lemma:expand},
the agent always begins DFTs at the same node.
We denote this node by $\vst$.
For any node $v$,
consider a shortest path $p = v_0, v_1, \dots, v_{\ell}$ from $\vst$ to $v$,
\ie $v_0 = \vst$ and $v_{\ell}=v$.
Let $Y$ be the maximum integer such that $v_Y$ has already been visited. Note that $Y$ is monotonically non-decreasing.
For each iteration of DFT,
we say that the iteration \emph{succeeds}
if by this iteration, $Y$ increases by at least $\lfloor (c-1)/2\rfloor$
or reaches to $\ell$.
The agent visits $v=v_{\ell}$ before
it performs $\chi = D/\lceil (c-1)/2 \rceil$ successful iterations of DFT.
Since the agent chooses a new color uniformly at random among
$c-1$ colors when it begins a new DFT, 
each iteration of DFT succeeds with a probability of at least $1/2$.
Suppose the agent repeats a DFT $2\chi + \lceil 16\log_e n\rceil$ times,
and let $Z$ be the number of successful iterations among those DFTs.
By Chernoff bound, $Z \ge \ex[Z]/2\ge \chi$
holds with a probability
of at least $1-\exp(-\ex[Z]/8)\ge 1-\exp(-2\log_e n)=1-n^{-2}$.
Thus, the agent visits any node $v$ in $O(D/c+\log n)$ DFTs
with probability $1-n^{-2}$.
Thus, $X = O(D/c+\log n)$ with high probability,
yielding $\ex[X] = O(D/c+\log n)$.
 \end{proof}


\section{Deterministic Exploration}
\label{sec:det}
The randomized algorithm
$\pran{c}$, which is presented in the previous section,
achieves self-stabilizing exploration in a small number of steps in expectation. The key idea is to repeat DFTs with changing the color of the agent randomly. Unfortunately, we cannot make use of the same idea if we are not allowed to use randomization
and we choose a new color deterministically among $c$ candidates. For example, consider an $(n/2,n/2)$-lollipop graph $G_L=(V_L,E_L)$, which consists of a clique $K_{n/2}$ of $n/2$ nodes and a path $P_{n/2}$ of $n/2$ nodes, connected by a bridge edge. Let $u_1, u_2, \dots, u_{n/2}$ be the nodes in $P_{n/2}$ such that one endpoint $u_1$ neighbors one node of clique $K_{n/2}$. Suppose that the agent begins a new DFT in a configuration where it is located in the clique, and all the nodes in the clique have the same color. (Figure \ref{fig:lolli}.)
For $i=1,2,\dots,n/2$, let $c_i$ be the colors that the agent (deterministically) chooses for the $i$-th DFT.
Then, the adversary can assign color $c_i$ to each node $u_i$ in the path in the initial configuration. In this case, the deterministic version of $\pran{c}$ requires $\Omega(n^3)$ steps to visit all the nodes in $G_L$
because the agent visits $u_i$ in the $i$-th DFT but it immediately backtracks without visiting $u_{i+1}$,
and each DFT involves the move through every edge in the clique,
resulting in $\Theta(n^2)$ steps.

\begin{figure}
  \centering
  \includegraphics[width=0.7\linewidth]{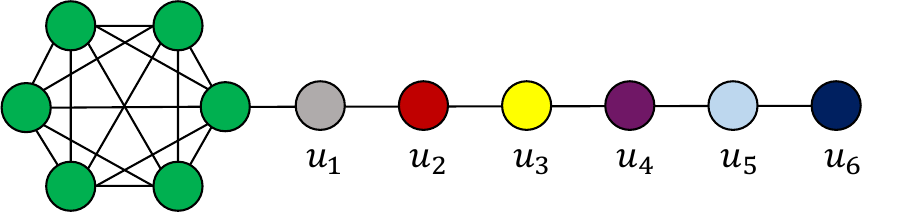}
  \caption{A (6,6)-lollipop graph}
 \label{fig:lolli}
\end{figure}
 
In this section, 
we give a deterministic
self-stabilizing exploration algorithm $\pdet{k}$
with design parameter $k$.
We must satisfy $k \ge \max(D,\dmax)$,
otherwise this algorithm no longer guarantees the correctness.
Thus we require the knowledge of 
an upper bound on $\max(D,\dmax)$.
Unlike $\pran{c}$,
the cover time of this algorithm does not depend on parameter $k$:
the agent visits all nodes within $O(m+nD)$ steps starting from any configuration.
Algorithm $\pdet{k}$ requires $O(\log k)$ bits both in the agent memory
and the whiteboard of each node $v \in V$.
Throughout this section,
we denote by $\dist(u,v)$ the distance between two nodes $u$ and $v$,
\ie the length of a shortest path from $u$ to $v$.

Whereas $\pran{c}$ is based on depth first traversals, $\pdet{k}$ is based on Breadth First Traversals (BFTs). Specifically, the agent attempts to construct a Breadth First Search (BFS) tree rooted at some node $r \in V$. Constructing a BFS tree by a single agent in a self-stabilizing setting poses several challenges, especially when bounding the time complexity by $O(m+nD)$ steps. In the rest of this section, we first introduce some terminologies and then explain the challenges posed by self-stabilization and how we address them.



In the remainder of this section, we present all key ideas of $\pdet{k}$, which are sufficient to verify the proof of our main theorem. A more detailed description of $\pdet{k}$, including its pseudocode, is provided in the appendix (Section \ref{sec:implement}), which may be helpful for gaining a deeper understanding of $\pdet{k}$.

\begin{figure}[t]
  \centering
  \includegraphics[width=0.5\linewidth]{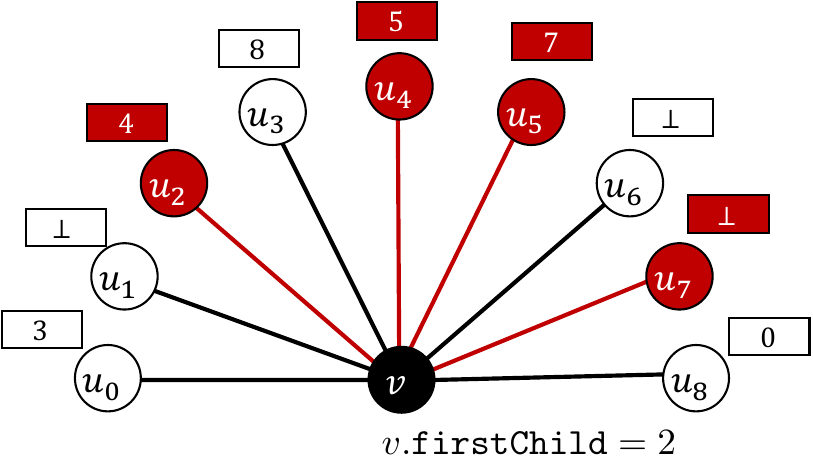}
  \caption{A valid node $v$ with $N(v)=\{u_0,u_1,\dots,u_8\}$.
  The index $i$ of each $u_i$ represents the corresponding number in $v$,
  \ie $u_i=N(v,i)$. The number in each rectangle represents
  the value of $u_i.\nextchild$.  
  In addition, $v = N(u_i,u_i.\parent)$ and $u_i.\level = v.\level = 1$
  must hold for all $i\in\{2,4,5,7\}$.
  Then, $u_2$, $u_4$, $u_5$, and $u_7$ are the children of $v$.
  }
 \label{fig:validNode}
\end{figure}

\subsection{Consistent Tree}
In $\pdet{k}$, the agent attempts to maintain a BFS tree that it can periodically traverse. Due to memory constraints, explicitly storing this tree structure at the nodes is non-trivial: the set of children of a node $v$ cannot be stored at $v$, as only $O(\log k)$ bits are available on $v$'s whiteboard. In this subsection, we explain how to embed a tree within the graph while adhering to these memory constraints. The main idea is as follows: since any tree has at most $n - 1$ parent-child relationships, we can store this information in a distributed manner across the $n$ nodes, with each node storing only $O(\log k)$ bits. Our approach is inspired by the well-known ``left-child, right-sibling'' representation. 

Each node $v \in V$ maintains the following whiteboard variables:
\begin{align*}
&v.\parent \in \{0,1,2,\dots,\delta_v-1\}, \quad v.\firstchild, v.\nextchild \in \{\bot, 0,1,\dots,k-1\},\\
&v.\level \in \{0,1,\dots,k\}.
\end{align*}
A node $v$ is \emph{valid} if $v.\firstchild = \bot$ or there is a sequence of nodes $c_1, c_2, \dots, c_{\ell}$ belonging to $N(v)$ with $\ell \geq 1$ that satisfies the following conditions:
\begin{itemize}
    \item $c_i.\parent = v$ and $c_i.\level = v.\level + 1$ for all $i \in [1, \ell]$,
    \item $p_v(\{v, c_i\}) < p_v(\{v, c_{i+1}\})$ and $c_i.\nextchild = p_v(\{v, c_{i+1}\})$ for all $i \in [1, \ell-1]$,
    \item $N(v,v.\firstchild)=c_1$, and
    \item $c_{\ell}.\nextchild = \bot$.
\end{itemize}
We call those nodes $c_1, c_2, \dots, c_{\ell} \in N(v)$ the children of $v$.
A typical example of a valid node is provided in Figure \ref{fig:validNode}.

A tree $T = (V_T, E_T)$ rooted at $r \in V$ is \emph{consistent} if $r.\level = 0$, all nodes in $V_T$ are valid, and for any $u,v \in V$, $u$ is a child of $v$ if and only if $(u, v) \in E_T$ holds.
By definition, we observe the following proposition.
\begin{observation}
\label{obs:uniqe_tree}
For any configuration $C$ and any node $v \in V$, there is at most one consistent tree that contains $v$ in $C$.
\end{observation}
We say that a node $v \in V$ is \emph{consistent} if there is a consistent tree $T=(V_T,E_T)$ such that $v \in V_T$.
Observation \ref{obs:uniqe_tree} implies that the consistent tree that contains a consistent node $v$ is uniquely determined.
We denote such a tree by $T(v)$. Notice $T(v)$ may not be rooted at $v$.

In the remainder of this paper, for any tree $T=(V_T,E_T)$, we use $v \in T$ and $e \in T$ to denote $v \in V_T$ and $e \in E_T$, respectively. We define $|T|$ as the number of nodes in the tree, i.e., $|T| = |V_T|$.


Note that a consistent tree $T=(V_T,E_T)$ is not necessarily a subgraph of a BFS tree. For example, there may be some node $v \in V_T$ for which $v.\level > \dist(r,v)$. Additionally, there could be a node $u \notin V_T$ such that $\dist(r,u) < \dist(r,w)$ for some $w \in V_T$.

\subsection{Tree Expansion}
\label{sec:expand}
Our deterministic algorithm $\pdet{k}$ expands $T(\vcur)$ by incorporating new nodes one by one. The first challenge is determining how to expand $T(\vcur)$. Consider the scenario where the agent is located at a consistent node $v \in V$, and explores an edge $e=\{v,u\} \notin T(v)$. If $u \notin T(v)$, the agent should incorporate the destination node $u$ into $T(v)$
by setting $u.\parent \gets p_u(e)$ and $u.\level = v.\level+1$,
initializing $u.\firstchild \gets \bot$,
and updating $u.\nextchild$ and $v.\firstchild$ appropriately.
However, the agent must refrain from making these changes if $u$ is already part of $T(v)$.
A significant challenge arises here because when the agent visits node $u$, another node $w \in T(v)$ may be in the same state as $u$, complicating the determination of whether $u$ belongs to $T(v)$.
\begin{figure}
  \centering
  \includegraphics[width=\linewidth]{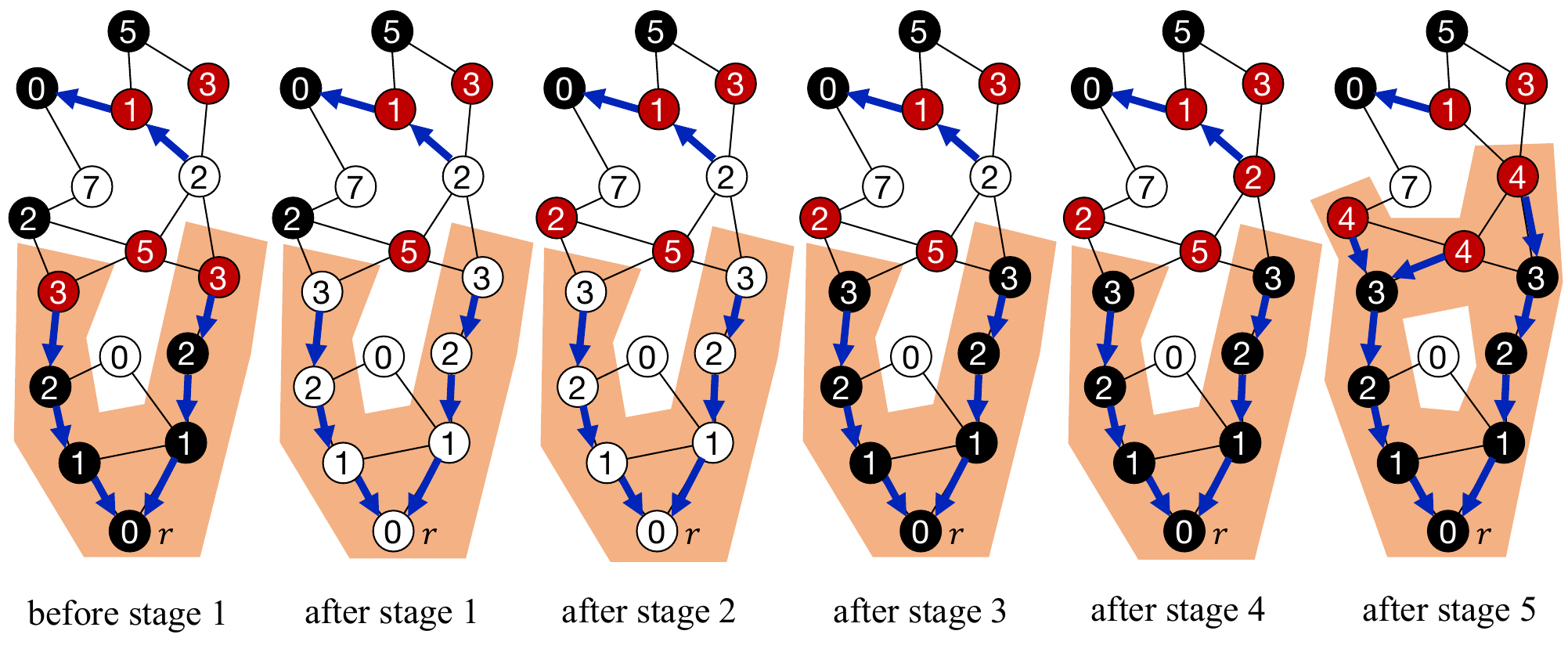}
  \caption{How to expand tree $T(r)$ in phase $3$. A number in a circle represents the level of the corresponding node.
  An arrow from node $u$ to $v$ indicates that
  $u$ is a child of $v$.
  The orange region represents which nodes join $T(r)$.
  }
 \label{fig:expand}
\end{figure}

We resolve this challenge by maintaining two agent variables $\self.\phase \in \{0,1,2,\dots,k-1\}$ and $\self.\stage \in \{1,2,3,4,5\}$, and one whiteboard variable $v.\clr \in \{B,W,R\}$ for each $v \in V$. Here, $B$, $W$, and $R$ stand for \emph{black}, \emph{white}, and \emph{red}, respectively. When $\self.\phase = i$, the agent is said to be in phase $i$. Suppose the agent begins phase 0 at node $r \in V$. In this phase, the agent visits all neighbors of $r$, incorporating them such that $T(r)$ becomes the star centered at $r$, containing all nodes in $N_1(r)$.
For any $i \ge 1$, let $V(r,i)$ be the set of nodes with level $i$ in $T(r)$, and $\vtarget(r,i)$ be the set of nodes that are neighbors to a node in $V(r,i)$ and are not included in $T(r)$. The objective of phase $i$ is to incorporate all nodes in $\vtarget(r,i)$ into $T(r)$. Each phase consists of five sub-phases, or \emph{stages}, managed by the variable $\self.\stage \in \{1,2,3,4,5\}$, described as follows.
(Figure \ref{fig:expand} illustrates a typical example.)
\begin{itemize}
\item \textbf{Stage 1 (Coloring the Current Tree White)}: The agent circulates within $T(r)$ (\ie{} visits all nodes in $T(r)$ in $2|T(r)|-2$ steps), coloring all nodes white.
\item \textbf{Stage 2 (Coloring the Black Targets Red)}: The agent recirculates within $T(r)$. During the circulation, each time it visits a node $v$ at level $i$, it visits each neighbor $u \in N(v)$ and colors $u$ red if $u$ is black.
\item \textbf{Stage 3 (Coloring the Current Tree Black)}: The agent recirculates within $T(r)$, this time coloring all nodes in $T(r)$ black.
\item \textbf{Stage 4 (Coloring the White Targets Red)}: The agent recirculates within $T(r)$. Similar to Stage 2, as it visits each node $v$ at level $i$, it visits all of $v$'s neighbors, this time changing the color of white nodes to red.
\item \textbf{Stage 5 (Incorporating the Targets)}: In the final circulation of $T(r)$, similar to Stage 2, the agent visits all nodes that neighbor a node at level $i$ in $T(r)$, incorporating any red nodes it encounters into $T(r)$.
Thereafter, it updates $\self.\phase \gets (\self.\phase+1) \mod k$.
\end{itemize}
During Stages 1 and 2, all nodes in $\vtarget(r,i)$ that are initially black are colored red, and during Stages 3 and 4, all initially white nodes are turned red. By the beginning of Stage 5, all nodes in $\vtarget(r,i)$ are red, and all nodes in $T(r)$ are black, thus allowing the agent to clearly distinguish and incorporate nodes from $\vtarget(r,i)$ into $T(r)$.

In every stage, the agent circulates $T(r)$ exactly once. This circulation is easily accomplished as follows:
Here, we refer to the movement from a parent to a child as a \emph{forward move} and the movement from a child to a parent as \emph{backtracking}. Each time the agent makes a forward move and visits a node $u$, it moves to $u$'s first child $N(u,u.\firstchild)$ if $u.\firstchild \neq \bot$; otherwise, it simply returns to $u$'s parent.
Suppose the agent backtracks from a node $u$ to $u$'s parent, say $v$, and $u$ is the $i$-th child of $v$. Before backtracking, the agent tentatively memorizes $u.\nextchild$ in its own variable $\self.\nextchild$. Thus, after backtracking, it can go to the $i+1$-th child of $v$ if one exists. If $u$ is the last child of $v$, $\self.\nextchild$ must be $\bot$ at that time.
Then, the agent returns to $v$'s parent, or terminates the circulation if $v = r$.
This process clearly circulates $T(r)$, moving through every edge in the tree exactly twice, requiring $2|T(r)| - 2$ steps.


In Stages 2, 4, and 5, each time the agent visits a node $v$ at level $\self.\phase$, it must visit all of $v$'s neighbors. At this point, the agent moves to the neighbors in descending order of the port numbers; specifically, it visits $u_{\delta_v-1}, u_{\delta_v-2}, \dots, u_0$ in this order, where $u_i = N(v,i)$ for $i \in [0,\delta_v-1]$.
This ordering helps the agent efficiently incorporate red neighbors of $v$ into $T(r)$ during Stage 5. Initially, the agent sets both $v.\firstchild$ and $\self.\nextchild$ to $\bot$. Thereafter, each time the agent visits a red neighbor $u_i$, it first updates $u_i.\nextchild \gets \self.\nextchild$ in addition to $u_i.\level = v.\level + 1$, $u_i.\parent = \pin$, and $u_i.\firstchild \gets \bot$. After returning to $v$, it updates $\self.\nextchild \gets i$ and $v.\firstchild \gets i$.
Once this process is completed, all red nodes in $N(v)$ are incorporated into $T(r)$, ensuring that $T(r)$ remains consistent.

\begin{lemma}
\label{lemma:expand2}
If the agent begins phase $i$ at a consistent node $r$ with level $0$, $\vtarget(r,i)$ is incorporated into $T(r)$ before the end of phase $i$. Phase $i$ ends in $O\left(n+\sum_{v \in V(r,i)}\delta_v\right)$ steps.
\end{lemma}
\begin{proof}
The correctness of the first part of the lemma, concerning the incorporation of $\vtarget(r,i)$ into $T(r)$ in phase $i$, follows from the above discussion.
We derive the second part from the fact that:
(i) each of the five circulations on $T(r)$ requires only $2|T(r)|-2=O(n)$ steps, and
(ii) in Stages 2, 4, and 5, for each node $v \in T(r,i)$, the agent moves from $v$ to each $u \in N(v)$ and vice versa, which requires exactly $3 \cdot 2 \cdot \sum_{v \in V(r,i)}\delta_v$ steps in total.
 \end{proof}

\begin{corollary}
\label{cor:after_reset}
Once the agent begins phase $0$,
it visits all nodes within $O(m+nD)$ steps.
\end{corollary}
\begin{proof}
Let $r$ be the node at which the agent begins phase $0$. As defined, phase $0$ completes within $O(\delta_r)$ steps. During the subsequent phases $1, 2, \dots, D-1$, by Lemma \ref{lemma:expand2}, all nodes are incorporated into $T(r)$ within $O(m+nD)$ steps in total. 
 \end{proof}

\begin{figure}[t]
  \centering
  \includegraphics[width=0.5\linewidth]{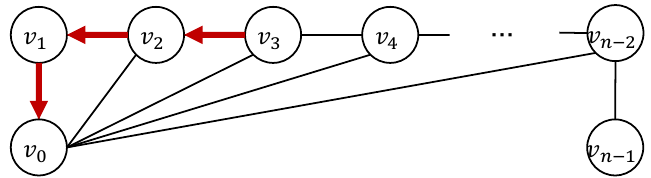}
  \caption{A consistent but undesirable tree}
 \label{fig:worstcase}
\end{figure}

\subsection{Towards Fast Reset}
\label{sec:fast_reset}

Since we seek a self-stabilizing solution, we must address arbitrary initial configurations where the current node may not be consistent, or $T(\vcur)$ significantly deviates from a BFS tree structure. How can we ensure that the agent visits all nodes within $O(m+nD)$ steps starting from such a configuration? According to Corollary \ref{cor:after_reset}, this is feasible if the agent can reset $\self.\phase$ to zero and initiate phase $0$ within every $\Theta(m+nD)$ steps with a sufficiently large hidden constant. However, the agent does not know the values of $n$, $m$, or $D$; it only knows an upper bound $k \ge \max(D,\dmax)$, which may significantly exceed $D$, i.e., $k = \omega(D)$. Therefore, we require a sophisticated mechanism to reset $\self.\phase$ to zero within $O(m+nD)$ steps from any configuration, without harming tree expansion so that the validity of Corollary \ref{cor:after_reset} is preserved.

It is straightforward to handle configurations where the current node, $\vcur$, is not consistent. As mentioned earlier, the agent attempts to circulate the consistent tree $T(\vcur)$ in every stage of each phase, assuming that $\vcur$ is consistent. If $\vcur$ is not consistent, the agent will encounter some inconsistency within $O(m)$ steps during the circulation process, or $\vcur$ becomes consistent. For instance, an inconsistency is detected when the agent moves from a node $v$ to $u$, one of $v$'s children, but either $u.\parent \neq p_u(\{u,v\})$ or $u.\level \neq v.\level + 1$ holds.
The agent begins phase $0$ whenever it detects any such \emph{local} inconsistency.
(See \ref{sec:implement} for detailed rules on detecting inconsistencies).
Hence, in the rest of this section,
we assume that $\vcur$ is consistent without loss of generality. 

The main challenge is handling configurations where $\vcur$ is consistent, yet $T(\vcur)$ significantly deviates from a BFS tree structure. To illustrate the difficulty of this scenario, consider the configuration shown in Figure \ref{fig:worstcase}, where the agent begins phase 3 at a node $r = v_0$ and $T(r) = \{(v_3, v_2), (v_2, v_1), (v_1, v_0)\}$. In this configuration, edges $\{v_0, v_1\}, \{v_0, v_2\}, \dots, \{v_0, v_{n-2}\}$ are never utilized in subsequent tree circulations, as $v_1$ is the only child of $v_0$. Consequently, exactly one node is added to $T(r)$ in each phase, requiring $\Omega(n^2)$ steps to visit all nodes and reset $\self.phase$ to zero, despite $D = 2$
and $m=O(n)$.

If we ignore time complexity, this challenge can be easily addressed. 
For example, consider what happens if we require $T(\vcur)$ to have some kind of BFS tree structure in addition to being consistent. Specifically, let $r$ be the root of $T(\vcur)$ and suppose we require:
(i) for any $v \in T(r)$, $v.\level = \dist(r,v)$,
and (ii) all nodes $v \in N_H$ are contained in $T(\vcur)$, where $H$ is the height of $T(r)$, \ie $H=\max\{u.\level \mid u \in T(r)\}$. 
Then, the agent can verify these conditions in each circulation of $T(\vcur)$ by visiting all neighbors of each node in $T(\vcur)$ and reset $\self.\phase$ to zero if it finds any violations of these conditions.
However, this may require $\Omega(m)$ steps in each circulation, resulting in $\Omega(mD)$ steps in the entire execution of the algorithm.

We address the second challenge (\ie deviation from the BFS structure) by adding the following three rules:
\begin{itemize}
\item \textbf{Rule 1}: 
In Stage 5 of each phase, the agent checks whether it has incorporated at least one node into $T(\vcur)$. If no node has been incorporated, indicating that $T(\vcur)$ can no longer be expanded, the agent resets $\self.\phase$ to zero.
\item \textbf{Rule 2}: In Stage 1 of each phase $i$, 
the agent resets $\self.\phase$ to zero if it finds a node $v \in V(r,i)$ such that $v.\firstchild \neq \bot$.
\item \textbf{Rule 3}: As explained above, in Stage 5 of each phase $i$, the agent visits all neighbors of each node $v \in V(r,i)$.
The agent resets $\self.\phase$ to zero if it finds a black neighbor $u \in N(v)$ with $u.\level < \self.\phase - 1 = i -1$.
\end{itemize}
In what follows, we first observe that these three rules do not invalidate Corollary \ref{cor:after_reset}, and then demonstrate that these rules ensure a high frequency of resetting $\self.\phase$.

\begin{observation}
\label{obs:alive}
Corollary \ref{cor:after_reset} still holds even with the addition of the above three rules.
\end{observation}

\begin{proof}
Suppose the agent begins phase 0 at a node $r$. The first rule is applicable only after the agent has incorporated all nodes into $T(r)$ and thus does not affect the corollary.
The second rule is never triggered because the tree expansion proceeds hop by hop after initiating phase 0.
The third rule is also never applicable in the execution subsequent to resetting $\self.\phase$ to zero, because at the beginning of each phase $i$, $T(\vcur)$ forms a shortest path (\ie BFS) tree over $G[N_i(r)]$. This configuration ensures that the difference in levels between any two nodes $u, v \in V'$, where $V'$ is the set of nodes in $T(\vcur)$ at the start of phase $i$, satisfies $|u.\level - v.\level| \le 1$. Moreover, all black nodes that the agent visits in Stage 5 must be within $V'$: all neighbors of the nodes in $V(r,i)$ that are not part of $V'$ are colored red in either Stage 2 or 4.
Thus, the condition of the third rule is never satisfied.
 \end{proof}


\begin{lemma}
\label{lemma:begins_zero}
Starting from any configuration, the agent begins phase 0 in $O(m+nD)$ steps.
\end{lemma}

\begin{proof}
As mentioned above, the agent begins phase 0 within $O(m)$ steps by finding a local inconsistency if the current node is not consistent. Thus, we assume the current node is consistent. Since each phase requires only $O(m)$ steps, starting from any configuration, the agent initiates some phase $i$ at some node, say $\vst$, within $O(m)$ steps. By the second rule, we can assume without loss of generality that at that time all nodes $v \in V(\vst,i)$ do not have children in $T(\vst)$, \ie $v.\firstchild = \bot$; otherwise, the agent begins phase 0 in the next $O(m)$ steps.

In the rest of this proof, we show that the agent resets $\self.\phase$ to zero before it enters phase $i + 2D$, thereby supporting the lemma by Lemma \ref{lemma:expand2}. Assume for contradiction that it enters phase $i + 2D$ and focus on the configuration $C$ at that time. Let $V_{< i} = \bigcup_{j < i} V(r,j)$, $u$ any node in $V(r,i)$, and $v$ any node in $V(r,i+D+1)$. Note that the first rule guarantees $V(r,i+D+1) \neq \emptyset$. After initiating phase $i$, the agent expands $T(r)$ hop by hop, avoiding the region of $V_{<i}$. This implies that the distance between $u$ and $v$ in the induced subgraph $G[V \setminus V_{<i}]$ is at least $D+1$; otherwise, $v.\level \le i+D$ holds, contradicting $v \in V(r,i+D+1)$. Since $\dist(u,v) \le D$, a shortest path $v_0, v_1, \dots, v_{\ell}$ between $v$ and $u$ must contain a node in $V_{<i}$, where $v_0 = v$, $v_{\ell} = u$, and $\ell \le D$ holds. Let $s$ be the smallest index such that $v_s \in V_{<i}$, \ie $v_0, v_1, \dots, v_{s-1} \notin V_{<i}$. Since $v=v_0$ and $v_{s-1}$ belong to the same connected component in $G[V \setminus V_{<i}]$, $v_{s-1}$ is contained in $T(r)$ and $v_{s-1} \in V(r,i')$ holds for some $i' \in [i+1, i+D+s]$ in configuration $C$. Note that we have $i' \neq i$ here because otherwise $v.\level \le i+s-1 < i+D+1$ holds, a contradiction with the fact $v \in V(r,i+D+1)$. Therefore, in Stage 5 of phase $i' \le i+D+s < i+D+\ell \leq i+2D$, the agent finds a black node with a level at most $i-1$,
\ie it visits $v_s \in V_{<i}$.
Since $i' \ge i+1$, the agent detects $v_s.\level < i' - 1$ and triggers the third rule, resetting $\self.\phase$ to zero, a contradiction.
 \end{proof}

The following theorem is derived from Observation \ref{obs:alive} and Lemma \ref{lemma:begins_zero}.

\begin{theorem}
\label{theorem:pdet}
Algorithm $\pdet{k}$ solves the self-stabilizing exploration problem for all simple, undirected, and connected graphs with a diameter and maximum degree of at most $k$. The cover time is $O(m+nD)$ steps, regardless of the value of $k$. The memory requirement is $O(\log k)$ for both the agent and each node.
\end{theorem}

\section{Concluding Remarks}
In this paper, two self-stabilizing graph exploration algorithms---a deterministic algorithm and a randomized algorithm---for a single mobile agent are presented. The randomized algorithm $\pran{c}$, parameterized by $c \ge 2$, achieves a cover time of $\rctime$, while requiring $O(\log c)$ bits of agent memory and $O(\log (c + \delta_v))$ bits per node $v$. This algorithm is time-optimal when $c = \Omega(n)$. The deterministic algorithm $\pdet{k}$ achieves a cover time of $O(m + nD)$ and uses $O(\log k)$ bits of memory for both the agent and each node.
This algorithm requires that the parameter $k$ satisfies $k \ge \max(D, \dmax)$; in other words, the agent must have prior knowledge of upper bounds on the diameter $D$ and maximum degree $\dmax$.

An interesting research direction is to investigate self-stabilizing exploration with multiple agents. When two or more agents are involved, the cover time can be defined as the maximum (expected) number of time steps required until every node is visited by at least one agent, starting from any configuration. A natural question arises: Can multiple agents reduce the cover time? If so, can even a constant number of agents achieve a significant reduction, or does the improvement in cover time depend explicitly on the number of agents?

\section*{Acknowledgements}
This work was supported by JST FOREST Program JPMJFR226U and 
JSPS KAKENHI Grant Numbers JP19K11826, JP20KK0232, JP22K11903, JP23K11059, and JP23K28037.

  \bibliographystyle{plain} 
  \bibliography{agent}
  
\clearpage

\appendix
\section{Lower Bound}
\label{sec:lower}
By definition, the cover time of any exploration algorithm is
trivially $\Omega(n)$.
In this section, we give a better lower bound:
the cover time of any (possibly randomized) algorithm is $\Omega(m)$. 
Specifically, we prove the following theorem. 
\begin{theorem}
\label{theorem:lower}
Let $\calp$ be any exploration algorithm.
For any positive integers $n,m$ such that $n \ge 3$ and $n-1 \le m \le n(n+1)/2$,
there exits a simple, undirected, and connected graph $G=(V,E)$
with $|V| = n$ and $|E| = m$
such that
the agent running $\calp$ on $G$
starting from some node in $V$
requires at least $(m-1)/4$ steps
to visit all nodes in $V$ in expectation.
\end{theorem}

\begin{proof}
For simplicity, we assume that $\calp$ is
a randomized algorithm. However, the following proof also
holds without any modification
even if $\calp$ is a deterministic one.

Let $G'=(V',E')$ be any simple, undirected, and connected graph
such that $|V'|=n-1$ and $|E'|=m-1$. There must be such a graph
because $n-1 \ge 2$ and $m-1 \ge (n-1)-1$.
Suppose that the agent runs $\calp$ on $G'$
starting from any node $\vst \in V'$.
Let $X_{u,v}$ be the indicator random variable such that
$X_{u,v}=1$ if the agent traverses edge $\{u,v\}$
(from $u$ to $v$ or from $v$ to $u$)
in the first $(m-1)/2$ steps,
and $X_{u,v}=0$ otherwise.
Clearly, the agent traverses at most $(m-1)/2$ edges 
in the first $(m-1)/2$ steps.
Therefore, $\sum_{\{u,v\} \in E'} X_{u,v} \le (m-1)/2$
always hold,
thus $\sum_{\{u,v\} \in E'} \ex[X_{u,v}] = \ex[\sum_{\{u,v\} \in E'} X_{u,v}]\le (m-1)/2$. This yields that there exists
an edge $\{u,v\} \in E'$ such that
$\Pr[X_{u,v}=1] = \ex[X_{u,v}] \le 1/2$.

Let $\{u,v\} \in E'$ such that $\Pr[X_{u,v}=1] \le 1/2$
and define $G=(V,E)$ as the graph that we obtain
by modifying $G'$ as follows:
remove edge $\{u,v\}$,
introduce a new node $w \notin V'$,
and add two edges $\{w,u\}$ and $\{w,v\}$.
Formally, $V=V'\cup \{w\}$ and $E=E' \cup \{\{w,u\},\{w,v\}\}\setminus \{\{u,v\}\}$.
By definition, graph $G$ is simple, undirected, and connected.
Then, the agent running $\calp$
on $G$ starting from node $\vst$
requires at least $(m-1)/4$ steps in expectation
to visit all nodes in $V$
because the agent does not visit node $w$ in the first $(m-1)/2$
steps with probability at least $1/2$.
\end{proof}

\section{Detailed Description of Deterministic Algorithm}
\label{sec:implement}

\setcounter{AlgoLine}{0}
\begin{algorithm}
\caption{Main Routine of $\pdet{k}$ and function $\initialize()$}
\label{al:pdet_main}
\VarAgent{
\\
\varspace $\self.\phase,\self.\level \in  \{0,1,\dots,k-1\}$,\\
\varspace $\self.\stage \in \{1,2,3,4,5\}$,\\
\varspace $\self.\nextchild \in \{\bot, 0,1,\dots,k-1\}$,\\
\varspace $\self.\error,\self.\found \in \{\fl,\tr\}$,\\
\varspace $\self.\mode \in \{-1,1\}$
}
\aspace

\VarNode{
\\
\varspace $v.\parent \in \{0,1,\dots,\delta_v-1\}$, \\
\varspace $v.\firstchild, v.\nextchild \in \{\bot, 0,1,\dots,k-1\}$,\\
\varspace $v.\clr \in \{W,B,R\}$, \\
\varspace $v.\level \in \{0,1,\dots,k\}$,\\
\varspace $v.\port \in \{-1,0,\dots,\delta_{v}\}$
}
\aspace 
\Note{
Whenever the agent decides to move via a port not in $[0,\delta_{\vcur}-1]$,
it immediately detects the error and jumps to Line 8.
}
\aspace

\While{$\tr$}{
\If{$\self.\phase = 0$}{
$\initialize()$; $\self.\stage \gets 1$\;
}
\If{$\self.\phase \ge 1$}{
$\circulate()$\;
$\self.\stage \gets (\self.\stage \bmod 5)+1$\;
}
\uIf{
$
\left(
\begin{aligned}
&\self.\error = \tr \lor (\self.\phase \ge 1 \\
&\land \self.\stage = 1 \land \self.\found = \fl) 
\end{aligned}
\right)
$
}{
$\self.\phase \gets 0$; $\self.\error \gets \fl$\;
}
\ElseIf{$\self.\stage = 1$}{
$\self.\phase \gets (\self.\phase + 1) \bmod k$\;
$\self.\found \gets \fl$\;
}
}
\end{algorithm}

\begin{algorithm}
\caption{Function $\initialize()$ in $\pdet{k}$}
\label{al:initialize}
\Fn{$\initialize()$}{
$\self.\level \gets 0$; $\vcur.\level \gets 0$;
$\vcur.\port \gets 0$\;
$\vcur.\firstchild \gets 0$\;
\While{$\vcur.\port < \delta_{\vcur}$}{
$\self.\nextchild \gets \begin{cases}
\vcur.\port + 1& \vcur.\port < \delta_{\vcur}-1\\    
\bot & \text{otherwise}
\end{cases}$\;
Migrate to $N(\vcur,\vcur.\port)$\;
$\vcur.\parent \gets \pin$;
$\vcur.\level \gets 1$\;
$\vcur.\firstchild \gets \bot$; $\vcur.\nextchild \gets \self.\nextchild$\;
Migrate to $N(\vcur,\pin)$\;
$\vcur.\port \gets \min(\vcur.\port + 1,\delta_{\vcur})$\;
}
}

\end{algorithm}

\begin{algorithm}
\caption{Function $\circulate()$ in $\pdet{k}$}
\label{al:circulate}

\Notation{\\
$
\begin{aligned} 
\varspace \nextd(v) = 
\begin{cases}
\self.\nextchild 
& \textbf{if } \self.\nextchild \neq \bot \lor v.\level = 0\\
v.\firstchild &
\textbf{else if } v.\parent = \pin \land v.\firstchild \neq \bot\\
v.\parent &
\textbf{otherwise}  
\end{cases}
\end{aligned}
$\\
\varspace $\colorupdate(a)=
\begin{cases}
W & \self.\stage = 1\\
B & \self.\stage = 3\\
a & \text{otherwise}		   
\end{cases}$
}
\aspace

\Fn{$\circulate()$}{
$\vcur.\clr \gets \colorupdate(\vcur.\clr)$\;
$\move(\vcur.\firstchild)$\;
\While{$\nextd(\vcur) \neq \bot \land \self.\error = \fl $}{
$\vcur.\clr \gets \colorupdate(\vcur.\clr)$\;
\uIf{$\self.\level = \self.\phase$}{
\If{$\self.\stage \in \{2,4,5\}$}{$\expand()$}
$\move(\vcur.\parent)$\;
}
\Else{

$\move(\nextd(\vcur))$\;
}
}
}
\end{algorithm}

\begin{algorithm}
\caption{Function $\move()$ in $\pdet{k}$}
\label{al:move}
\Fn{$\move(q)$}{
 $
\self.\mode \gets
\begin{cases}
-1 & v.\level > 0 \land q = \vcur.\parent
\\
1 & \text{otherwise}
\end{cases}
$\;
$\self.\level \gets \min(k,\max(0,\self.\level +\self.\mode))$\;
\uIf{$\self.\mode = 1$}{
 $\self.\nextchild \gets \bot$ 
 }\Else{
 $\self.\nextchild \gets \vcur.\nextchild$ 
 }
 Migrate to $N(\vcur,q)$\; 
 \If{
$
\begin{aligned}
&\self.\level \neq \vcur.\level\\
&\lor (\self.\nextchild \neq \bot \land \self.\nextchild \le \pin) \\
&\lor (\self.\mode = 1 \land \pin \neq \vcur.\parent)\\
&\lor (\self.\mode = -1 \land \pin = \vcur.\parent) 
\\
&\lor (\self.\stage = 1 \land \self.\phase = \vcur.\level \land \vcur.\firstchild \neq \bot)
\end{aligned}
$
}{
 $\self.\error \gets \tr$
 }
}
\end{algorithm}

\begin{algorithm}[t]
\caption{Function $\expand()$ in $\pdet{k}$}
\label{al:pdet_expand}

\Fn{$\expand()$}{
$\vcur.\port \gets \delta_{\vcur}-1$\;
$\vcur.\firstchild \gets \bot$; $\self.\nextchild \gets \bot$\;
\While{$\vcur.\port \ge 0$}{
\If{$\vcur.\port \neq \vcur.\parent$}{
Migrate to $N(\vcur,\vcur.\port)$\;
\uIf{
$
\left(
\begin{aligned}
&(\self.\stage = 2 \land \vcur.\clr = B)\\
& \lor (\self.\stage = 4 \land \vcur.\clr = W)
\end{aligned}
\right )
$
}{
$\vcur.\clr \gets R$; Migrate to $N(\vcur,\pin)$\;
}
\uElseIf{$\self.\stage = 5 \land \vcur.\clr = R$}{
$\vcur.\level \gets \self.\level+1$\;
 $\vcur.\parent \gets \pin$\;
 $\vcur.\firstchild \gets \bot$\;
 $\vcur.\nextchild \gets \self.\nextchild$\;
Migrate to $N(\vcur,\pin)$\;
$\self.\nextchild \gets \pin$\;
$\vcur.\firstchild \gets \pin$\;
$\self.\found \gets \tr$\;
}
\ElseIf{
$
\left (
\begin{aligned}
&\self.\stage = 5 \land \vcur.\clr = B \\
& \land \vcur.\level < \self.\phase - 1
\end{aligned}
\right )
$
}{
$\self.\error \gets \tr$\;
Migrate to $N(\vcur,\pin)$\;
}
}
$\vcur.\port \gets \max(-1,\vcur.\port-1)$\;
}
}
\end{algorithm}

The pseudocode of $\pdet{k}$ is presented in Algorithms \ref{al:pdet_main}, \ref{al:initialize}, \ref{al:circulate}, \ref{al:move}, and \ref{al:pdet_expand}.
Algorithm \ref{al:pdet_main} gives the list of variables
and the main routine, while Algorithms 
\ref{al:initialize}, \ref{al:circulate}, \ref{al:move}, and \ref{al:pdet_expand} give subroutines
$\initialize()$, $\circulate()$, $\move()$, and $\expand()$, respectively.

In addition to the variables we explained above (i.e., $\self.\phase$, $\self.\stage$, $\self.\nextchild$, $v.\parent$, $v.\firstchild$, $v.\nextchild$, $v.\level$, and $v.\clr$), the agent also maintains four additional variables: $\self.\level \in \{0,1,\dots,k\}$, $\self.\error \in \{\fl,\tr\}$, $\self.\found \in \{\fl,\tr\}$, and $\self.\mode \in \{-1,1\}$. Moreover, each node $v \in V$ maintains a whiteboard variable $v.\port \in \{0,1,\dots,\delta_v\}$.
We use $\self.\level$ to maintain the distance between the current node $\vcur$ and the root of $T(\vcur)$ in $T(\vcur)$.
We use $\self.\error$ and $\self.\found$ to reset $\self.\phase$ to zero, as we will see later. Two variables $\self.\mode$ and $v.\port$ are just temporary variables:
$\self.\mode$ is used to remember the last movement is a forward move or not, 
and $v.\port$ is used to visit all neighbors of the current node in functions $\initialize()$ and $\expand()$.

The main routine is simply structured. 
In phase 0,
the agent executes function $\initialize()$ and sets $\self.\stage$ to $1$
(Lines 2-3).
In $\initialize()$, the agent simply initializes
$\self.\level$ and $v.\level$ for all $v \in N_1(\vcur)$
and makes the parent-child relationship
between $\vcur$ and each of its neighbors (Lines 13-21).
Thus, at the end of phase 0,
$T(\vcur)$ is the tree rooted at $\vcur$
such that every $u \in N(\vcur)$ is a child of $\vcur$
and no node outside $N_1(\vcur)$ is included in $T(\vcur)$.
In phase $i \ge 1$,
it executes function $\circulate()$
and increments $\self.\stage$ by one;
If $\self.\stage$ is already $5$ before incrementation,
$\self.\stage$ is reset to $1$ (Lines 4-6).
We will explain $\circulate()$ later.
Every time $\self.\stage$ is reset to $1$,
the agent begins the next phase,
that is, increments $\self.\phase$ by one modulo $k$ (Line 10).
 To implement the first rule (Rule 1 in Section \ref{sec:fast_reset}),
 the agent always memorizes in $\self.\found$ whether
 it added a new node to $T(\vcur)$ in the current phase. 
If $\self.\found$ is still $\fl$ at the end of a phase except for phase 0, the agent resets $\self.\phase$ to 0 and begins phase 0 (Line 8). In addition, the agent always substitutes $\tr$ for $\self.\error$
when it finds any inconsistency in $\circulate()$.
This variable $\self.\error$ is also used for implementing the second and third resetting rules (Rules 2 and 3 in Section \ref{sec:fast_reset}).
Once $\self.\error = \tr$ holds, the agent immediately halts $\circulate$ (Line 25)
and resets $\self.\phase$ to 0 and begins phase 0 (Line 8).


When the agent invokes $\circulate()$ at a node $r$ in phase $i$, it circulates through all nodes in $T(r)$ unless it detects any inconsistency, in which case it sets $\self.\error \gets \tr$. This function, $\circulate()$, simply implements the movement strategy specified in Section \ref{sec:expand}: the agent always moves to $N(v,\nextd(v))$, unless $\nextd(v) = \bot$ (Line 32), where $\nextd(v)$ is defined in Algorithm \ref{al:circulate}.
Only exception is when the agent visits a node $v$ with
$v.\level = \self.\phase$. Then, it simply goes back to the parent node
(after invoking $\expand()$ in Stages 2, 4, and 5)
(Line 28--30).
The agent makes a move inside the tree by invoking $\move()$,
in which $\self.\level$ is adequately updated (Line 35).
In Stage 1 (resp.~3), it always changes the color of the visited node
to white (resp.~black).
In Stages 2, 4, and 5, whenever the agent visits a node with
$\level = \self.\phase$, it tries to add the neighbors
of the current node according to the strategy mentioned in Section \ref{sec:expand},
by invoking $\expand()$ (Line 29).

In $\move()$, every kind of local inconsistency regarding
the parent-child relationship is detected.
Specifically, the agent detects an inconsistency at the migration from $u$ to $v$ if
(i) $\self.\level \neq v.\level$,
(ii) $\self.\nextchild \neq \bot$, but $\self.\nextchild \le \pin$, 
(iii) the agent has made a forward move from $u$ to $v$, but $\pin \neq v.\parent$,
(iv) the agent has made backtracking from $u$ to $v$, but $\pin = v.\parent$.
If an inconsistency is detected, the agent sets $\tr$ for $\self.\error$ (Line 42).
Moreover, this function implements the second rule (Rule 2 in Section \ref{sec:fast_reset}):
it sets $\self.\error \gets \tr$ if the agent visits a node with the level $\self.\phase$ that has one or more children in Stage 1.

Function $\expand()$ (Lines 43-61) is a simple implementation of the strategy explained in Sections \ref{sec:expand}
and the third rule (Rule 3 in Section \ref{sec:fast_reset}).
\end{document}